\documentclass[12pt]{article}

\usepackage[margin=1in]{geometry}
\usepackage{setspace} 
\usepackage{amsmath,amsfonts,amssymb}
\usepackage[authoryear]{natbib}
\usepackage{url}
\usepackage{comment}
\usepackage{csquotes}
\usepackage{graphicx}
\usepackage{float}
\usepackage[toc,page]{appendix}

 \usepackage{amsthm}

\newtheorem{proposition}{Proposition}

\newtheorem{lemma}{Lemma}
\usepackage{enumitem}

\usepackage[colorlinks]{hyperref}
\hypersetup{colorlinks, citecolor=blue, filecolor=blue, linkcolor=blue, urlcolor=blue}

\usepackage{xcolor}
\definecolor{green}{rgb}{0.0, 0.5, 0.0}

\begin{document}

\title{Market Inefficiency in Cryptoasset Markets}
\date{} 

\author{
\begin{tabular}{cccc}
Joel Hasbrouck\thanks{email: jh4@stern.nyu.edu}
&
Julian Ma\thanks{email: julian.ma@ethereum.org}
&
Fahad Saleh\thanks{email: fahad.saleh@ufl.edu}
&
Caspar Schwarz-Schilling\thanks{email: caspar.schwarz-schilling@ethereum.org}\\
\emph{\small{NYU Stern}} & \emph{\small{Ethereum Foundation}} & \emph{\small{University of Florida}} &
\emph{\small{Ethereum Foundation}}
\end{tabular}
\vspace{0.3cm}
}

\maketitle

\begin{abstract}
\setstretch{1.2}
\noindent
We demonstrate market inefficiency in cryptoasset markets. Our approach examines investments that share a dominant risk factor but differ in their exposure to a secondary risk. We derive equilibrium restrictions that must hold regardless of how investors price either risk. Our empirical results strongly reject these necessary equilibrium restrictions. The rejection implies market inefficiency that cannot be attributed to mispriced risk, suggesting the presence of frictions that impede capital reallocation.
\end{abstract}
\vspace{1.0 cm}
\thispagestyle{empty}
~\\
\newpage
\setcounter{equation}{0}
\setcounter{footnote}{0}
\setcounter{page}{1}

\section{Introduction}

We demonstrate market inefficiency in cryptoasset markets while allowing for arbitrary risk pricing. Our approach compares investments with similar risk and derives equilibrium restrictions that must hold regardless of how investors price risk. We then test these restrictions empirically. We show that necessary equilibrium restrictions are strongly rejected, thereby demonstrating market inefficiency.

To add more detail, we examine three investments involving ether (ETH), the native asset of the Ethereum blockchain. The first investment is direct staking: investors stake ETH on the Ethereum blockchain and earn staking yields. The second investment is lending ETH through a Decentralized Lending Protocol (DLP) which pays interest to lenders. The third investment is liquid staking: investors deposit ETH with a Staking Service Provider (SSP), receive an asset backed by the deposited ETH known as staked ether (stETH), and lend this stETH on a DLP for interest. This third investment earns both the staking yield (passed through by the SSP) and the stETH lending yield. Because stETH is redeemable for ETH, it trades near parity with ETH, and thus all three investments are primarily exposed to ETH price risk. Nonetheless, liquid staking carries an additional risk: the stETH-ETH exchange rate can deviate from unity, a phenomenon known as de-pegging. This de-peg risk is small in magnitude but nonzero and we account for it within our framework by allowing that investors price de-pegging risk.

Given the similar risk profiles of the three referenced investments, we derive tight equilibrium relationships relating their yields. To add some intuition, if capital can flow freely among the three investments, then an investor lending ETH at a DLP can withdraw that capital and pursue liquid staking or direct staking. In turn, if one investment becomes more attractive, capital will flow toward it until yields adjust to restore indifference across investments. This implies that changes in one yield must be accompanied by corresponding changes in the others, giving us testable equilibrium restrictions on relative yields. All three investments share the dominant risk factor (ETH price), and while liquid staking adds de-peg risk, we allow investors to price this additional risk arbitrarily. Because the equilibrium restrictions we derive hold for any price of de-peg risk, a rejection of these restrictions cannot be explained by investor de-peg risk pricing.

Our empirical tests examine whether yields co-move as necessary for equilibrium. We find that yields do not move in a manner consistent with equilibrium. The stETH lending yield shows almost no response to changes in the ETH lending yield or the staking yield differential. Both equilibrium restrictions are strongly rejected. The failure of yields to co-move as implied by equilibrium indicates that capital is not flowing among these investments freely, implying market inefficiency. 

Our paper contributes to the literature on cryptoeconomics \citep{johnsaleh2025cryptoeconomics}, particularly the literature examining Ethereum's Proof-of-Stake. This literature begins with \citet{saleh2021blockchain}, which studies whether consensus arises in equilibrium. Several prominent works thereafter examine staking as an investment \citep{rocsu2021evolution, jermann2023macro, cong2025tokenomics, johnpowvspos, harvey2024productivity}. We extend that literature by abstracting from microeconomic details and developing general asset pricing restrictions. Using this framework, we are able to use the Ethereum ecosystem to test for market efficiency, highlighting that market efficiency fails.

\section{Institutional Details}
\label{sec:institutional}

This section provides institutional background on Ethereum staking, liquid staking, and decentralized lending. We focus on the mechanisms relevant to our model, particularly those that determine yields and the relationship between ETH and stETH.

\subsection{Ethereum Staking}

The Ethereum blockchain uses a Proof-of-Stake consensus mechanism to secure its ledger \citep{saleh2021blockchain, johnpowvspos}. In Proof-of-Stake, a participant known as a \emph{validator} locks capital in the form of ether (ETH) to participate in consensus. The Ethereum protocol requires a minimum of 32 ETH to become a validator. When selected by the protocol, validators propose blocks of transactions and attest to the validity of blocks proposed by others. In return for performing these duties, validators receive rewards from two sources \citep{john2025economics}. The first source is newly minted ETH, which corresponds to seigniorage. The protocol issues these rewards algorithmically, with the per-validator rate decreasing as more ETH is staked network-wide. The second source is transaction fees. Under Ethereum's EIP-1559 fee mechanism \citep{roughgarden2021transaction}, users pay a base fee (which is burned) plus a priority fee (which goes to validators). Validators also capture \emph{Maximal Extractable Value} (MEV), the profit available from reordering, inserting, or censoring transactions \citep{daian2020flash}.\footnote{Prominent examples of MEV include sandwich attacks \citep{park2023conceptual}, stale price arbitrage \citep{capponijia2025liquidity, lehar2025decentralized}, and JIT liquidity provision \citep{capponijiazhu}.} As an aside, in practice, specialized builders construct blocks to maximize MEV and bid for inclusion via proposer-builder separation mechanisms, transferring MEV to validators through priority fees \citep{schwarzschilling2023time, capponi2024pbs}.

The staking yield $\gamma^{ETH}_t$ in our model corresponds to the rate at which validators accumulate rewards from both seigniorage and priority fees (including MEV). As of late 2025, approximately 34 million ETH is staked, corresponding to an annualized staking yield of approximately 3\%. Importantly, the staking yield varies over time as the amount of staked ETH changes and as transaction fee revenue fluctuates with network activity. Our framework fully accommodates this time-variation, allowing that $\gamma^{ETH}_t$ evolves stochastically.

\subsection{Liquid Staking}

The 32 ETH minimum and the technical requirements of running a validator node present barriers to direct staking. \emph{Staking Service Providers} (SSPs) address these barriers by pooling capital from multiple users and operating validators on their behalf. Users deposit ETH with the SSP and receive a \emph{Liquid Staking Token} (LST) in return. The LST represents a claim on the underlying staked ETH plus accumulated rewards.

Lido is the largest SSP, managing approximately 9.5 million ETH (around 28\% of all staked ETH). Lido's LST is called stETH. Lido's operations are implemented via smart contracts, which are programs deployed on the Ethereum blockchain that execute automatically when called \citep{john2023smart}. When a user deposits ETH with Lido, the following occurs in a single atomic transaction: the user calls the \texttt{submit()} function on Lido's deposit contract,\footnote{Lido stETH contract: \url{https://etherscan.io/address/0xae7ab96520de3a18e5e111b5eaab095312d7fe84}} sending ETH to the contract; the contract mints stETH to the user's address at a one-to-one ratio. The deposited ETH is then allocated to node operators who stake it with the Ethereum protocol.

Lido charges a 10\% fee on all staking rewards, including both seigniorage and priority fees (and hence MEV). The remaining 90\% accrues to stETH holders. In our model, $\gamma^{ETH}_t$ denotes the full protocol staking yield (seigniorage plus priority fees), while $\gamma^{stETH}_t$ denotes the yield passed to stETH holders after Lido's fee. Thus $\gamma^{stETH}_t = 0.9 \times \gamma^{ETH}_t$, and $\gamma^{stETH}_t$ inherits the time-variation of the underlying protocol yield.

Because stETH represents a claim on ETH held by Lido, it trades near parity with ETH. Holders can redeem stETH for ETH by calling the \texttt{requestWithdrawals()} function on Lido's withdrawal contract.\footnote{Lido Withdrawal Queue contract: \url{https://etherscan.io/address/0x889edc2edab5f40e902b864ad4d7ade8e412f9b1}} The redemption is subject to a delay due to Ethereum's withdrawal queue, but once processed, holders receive ETH at the prevailing exchange rate within the protocol. This redeemability anchors the stETH price to ETH.

However, the stETH-ETH exchange rate can deviate from unity. We refer to such deviations as \emph{de-pegging}. In our model, the variable $\chi_t$ captures the log stETH-ETH price ratio, with $\chi_t = 0$ corresponding to parity and $\chi_t < 0$ corresponding to stETH trading at a discount.

\subsection{Decentralized Lending}

Decentralized Lending Protocols (DLPs) enable users to lend and borrow cryptoassets via smart contracts \citep{gudgeon2020defi}. Aave is the largest DLP on Ethereum, with over 40 billion USD in deposits. Users can supply assets to Aave's liquidity pools and earn interest from borrowers.

To supply ETH to Aave, a user calls the \texttt{supply()} function on Aave's pool contract,\footnote{Aave V3 Pool contract: \url{https://etherscan.io/address/0x87870bca3f3fd6335c3f4ce8392d69350b4fa4e2}} specifying the amount and the asset. The contract transfers ETH from the user and mints \emph{aTokens} (specifically, aWETH for ETH) to the user's address. These aTokens represent the user's claim on the supplied assets plus accrued interest. The interest rate paid to suppliers is determined algorithmically based on the utilization rate of the pool: when demand for borrowing is high relative to supply, the interest rate increases; when demand is low, the rate decreases \citep{rivera2023equilibrium}.

Aave also supports stETH as a supplied asset. Users holding stETH can supply it to Aave and earn the stETH lending yield $\psi^{stETH}_t$ in addition to the staking yield $\gamma^{stETH}_t$ that accrues to stETH holders. Similarly, users can supply ETH to earn the ETH lending yield $\psi^{ETH}_t$. The yields $\psi^{ETH}_t$ and $\psi^{stETH}_t$ in our model correspond to the supply rates offered by Aave for ETH and stETH, respectively. Because utilization rates fluctuate with market conditions, these lending yields vary dynamically. Our framework fully accommodates this time-variation, allowing that $\psi^{ETH}_t$ and $\psi^{stETH}_t$ may evolve stochastically.

All interactions with Aave occur through smart contracts, ensuring that the terms of lending are enforced programmatically. Borrowers must post collateral exceeding the value of their loans, and the protocol automatically adjusts interest rates based on market conditions. This transparency and automation allow us to observe lending yields directly from on-chain data.

\section{Economic Model}

We examine an infinite horizon setting where time is indexed by $t \in \{0, 1, 2, \ldots\}$ and $\mathcal{F}_t$ denotes the information available at time $t$. We examine three investment strategies available to ether (ETH) holders in each period: direct staking, lending via decentralized lending protocols, and liquid staking. We first describe the return structure of each strategy, then introduce the pricing framework, and finally derive equilibrium restrictions.

\subsection{Investment Opportunities}

\subsubsection{Direct Staking}

Investors holding ETH can earn yield by staking, which involves locking ETH to participate in Ethereum's proof-of-stake consensus protocol \citep{john2025economics}. The gross return from holding and staking ETH is:
\begin{equation}
R^{ETH,stake}_{t,t+1} = \exp\{\gamma^{ETH}_{t} - \kappa + r^{ETH}_{t,t+1}\}
\label{eqn:cumethreturn}
\end{equation}
where $r^{ETH}_{t,t+1}$ denotes the period-$t$ to period-$t+1$ log price return on ETH, $\gamma^{ETH}_{t}$ denotes the log staking yield over the same period, and $\kappa > 0$ denotes the cost of staking. The investor earns both the price appreciation and the staking yield, net of costs.

\subsubsection{Lending ETH}

Alternatively, investors can deploy ETH in decentralized lending protocols such as Aave. The gross return from holding and lending ETH is:
\begin{equation}
R^{ETH,lend}_{t,t+1} = \exp\{\psi_t^{ETH} + r^{ETH}_{t,t+1}\}
\label{eqn:cumethreturn-yf}
\end{equation}
where $\psi^{ETH}_{t}$ denotes the period-$t$ to period-$t+1$ log lending yield. The investor earns both the price appreciation and the lending yield.

\subsubsection{Liquid Staking}

Investors can obtain exposure to staking yield through liquid staking tokens (LSTs). To acquire stETH, an investor deposits ETH with a staking service such as Lido and receives stETH one-for-one. The staking service stakes the deposited ETH and passes a fraction of the staking yield to stETH holders. Additionally, the investor can lend their stETH in protocols like Aave. The gross return from this combined strategy is:
\begin{equation}
R^{stETH}_{t,t+1} = \exp\{\gamma^{stETH}_t + \psi^{stETH}_{t} + r^{stETH}_{t,t+1}\}
\label{eqn:cumstethreturn}
\end{equation}
where $\gamma^{stETH}_t$ denotes the log staking yield passed to stETH holders (net of the staking service's fee), $\psi^{stETH}_{t}$ denotes the log lending yield on stETH, and $r^{stETH}_{t,t+1}$ denotes the log price return on stETH.

The stETH price return is linked to the ETH price return through the stETH-ETH exchange rate:
\begin{equation}
    r^{stETH}_{t,t+1} = r^{ETH}_{t,t+1} + \chi_{t+1} - \chi_t
    \label{eqn:ethsteth}
\end{equation}
where $\chi_{t} := \log(P^{stETH}_t / P^{ETH}_t)$ is the log stETH-ETH price ratio. Thus, stETH returns inherit ETH price risk plus any changes in the exchange rate.

\subsubsection{Peg Risk}

The stETH-ETH exchange rate $\chi_t$ introduces additional risk. We model $\{\chi_t\}_{t=0}^\infty$ as a two-state Markov chain with states $\{0, -\eta\}$ where $\eta > 0$. The state $\chi_t = 0$ corresponds to parity (stETH trades at par with ETH), while $\chi_t = -\eta$ corresponds to a significant discount. We focus on two states because small deviations from parity are statistically indistinguishable from parity given market microstructure noise. The transition matrix under the physical measure is:
\[
P = \begin{bmatrix}
p_{0,0} & 1 - p_{0,0} \\
1 - p_{\eta,\eta} & p_{\eta,\eta}
\end{bmatrix}
\]
where $p_{0,0}$ is the probability of remaining at parity given the peg currently holds, and $p_{\eta,\eta}$ is the probability of remaining at a discount given the peg is currently broken.

\subsection{Asset Pricing Framework}

We allow for arbitrary prices of risk and specify the pricing kernel as:
\begin{equation}
    \Lambda_{t,t+1} = \exp\left( -r_{t,t+1} - \frac{1}{2}\lambda_{ETH}^2 - \lambda_{ETH} \varepsilon^{ETH}_{t+1} - \lambda_{\chi} \omega_{t+1} + \xi_t(\lambda_{\chi}) \right)
    \label{eqn:kernel}
\end{equation}
where $r_{t,t+1}$ denotes the risk-free rate, $\lambda_{ETH}$ is the market price of ETH risk, and $\lambda_{\chi}$ is the market price of peg risk. The standardized ETH return innovation is $\varepsilon^{ETH}_{t+1} := (r^{ETH}_{t,t+1} - \mu_t^{ETH})/\sqrt{v_t^{ETH}}$, where $\mu_t^{ETH} := \mathbb{E}[r^{ETH}_{t,t+1} | \mathcal{F}_t]$ and $v_t^{ETH} := \text{Var}[r^{ETH}_{t,t+1} | \mathcal{F}_t]$. We assume $\varepsilon^{ETH}_{t+1} | \mathcal{F}_t \sim N(0,1)$ to ease exposition. The peg state indicator is $\omega_{t+1} := \mathbf{1}_{\{\chi_{t+1} = -\eta\}}$, and we assume $\varepsilon^{ETH}_{t+1}$ and $\omega_{t+1}$ are conditionally independent given $\mathcal{F}_t$. The normalizing constant is:
\begin{equation}
    \xi_t(\lambda_{\chi}) := -\log \mathbb{E}[\exp\{-\lambda_{\chi} \omega_{t+1}\} | \mathcal{F}_t]
    \label{eqn:normalizing}
\end{equation}

We impose no arbitrage, implying:
\begin{equation}
    \mathbb{E}[\Lambda_{t,t+1}\, R^i_{t,t+1}~|~\mathcal{F}_t] = 1
    \label{eqn:euler}
\end{equation}
Equation \eqref{eqn:euler} is the fundamental asset pricing equation \citep{harrison1979martingales, duffie2001dynamic}. The pricing kernel $\Lambda_{t,t+1}$ reflects how investors value payoffs across different market outcomes: when $\Lambda_{t,t+1}$ is high, a dollar is worth more to investors than when $\Lambda_{t,t+1}$ is low. Our specification \eqref{eqn:kernel} incorporates two sources of risk: ETH price movements (through $\varepsilon^{ETH}_{t+1}$) and stETH de-pegging (through $\omega_{t+1}$). The parameters $\lambda_{ETH}$ and $\lambda_{\chi}$ govern how these risks are priced. Mechanically, $\Lambda_{t,t+1}$ decreases in $\lambda_{ETH} \varepsilon^{ETH}_{t+1}$: a larger $\lambda_{ETH}$ means the kernel assigns less weight to outcomes where $\varepsilon^{ETH}_{t+1}$ is high. Consequently, an asset whose returns are positively correlated with $\varepsilon^{ETH}_{t+1}$ receives a lower valuation, and must offer a higher expected return in compensation. The same concept applies to peg risk through $\lambda_{\chi}$. We leave both parameters unrestricted.

\subsection{Equilibrium Restrictions}

We now derive equilibrium restrictions by applying the pricing equation \eqref{eqn:euler} to each investment opportunity.

\begin{proposition}[ETH Staking]
\label{prop:eth-staking}
The expected ETH return satisfies:
\begin{equation}
    \mathbb{E}[r_{t,t+1}^{ETH}~|~\mathcal{F}_t] = r_{t,t+1} + \lambda_{ETH} \sqrt{v_{t}^{ETH}} - \gamma^{ETH}_{t} + \kappa - \frac{v_t^{ETH}}{2}
    \label{eqn:eth-staking}
\end{equation}
where $v_t^{ETH} := \text{Var}[r_{t,t+1}^{ETH}~|~\mathcal{F}_t]$.
\end{proposition}

The term $\lambda_{ETH} \sqrt{v_t^{ETH}}$ is the risk premium: investors require compensation for bearing ETH price risk. The staking yield $\gamma^{ETH}_t$ reduces the required price return (staking provides additional compensation), while staking costs $\kappa$ increase it. The term $v_t^{ETH}/2$ is a convexity adjustment.

\begin{proposition}[ETH Lending]
\label{prop:eth-lending}
The expected ETH return satisfies:
\begin{equation}
    \mathbb{E}[r_{t,t+1}^{ETH}~|~\mathcal{F}_t] = r_{t,t+1} + \lambda_{ETH} \sqrt{v_{t}^{ETH}} - \psi_t^{ETH} - \frac{v_t^{ETH}}{2}
    \label{eqn:eth-lending}
\end{equation}
\end{proposition}

The interpretation parallels Proposition \ref{prop:eth-staking}: the lending yield $\psi_t^{ETH}$ substitutes for the staking yield as an additional source of return.

\begin{proposition}[stETH Investment]
\label{prop:steth}
The expected ETH return satisfies:
\begin{equation}
    \mathbb{E}[r_{t,t+1}^{ETH}~|~\mathcal{F}_t] = r_{t,t+1} + \lambda_{ETH} \sqrt{v_t^{ETH}} + \tilde{\eta}_t - \gamma^{stETH}_t - \psi_t^{stETH} - \frac{v_t^{ETH}}{2}
    \label{eqn:steth}
\end{equation}
where $\tilde{\eta}_t$ is the risk-adjusted peg premium:
\begin{equation}
    \tilde{\eta}_t = \begin{cases}
        \log{\frac{1}{e^{-\eta} + \tilde{p}_{0,0} (1 - e^{-\eta})}} & \text{if } \chi_t = 0\\[6pt]
        \log{\frac{1}{e^{\eta} - \tilde{p}_{\eta,\eta} ( e^{\eta} - 1)}} & \text{if } \chi_t = -\eta
    \end{cases}
    \label{eqn:peg-premium}
\end{equation}
and the risk-adjusted transition probabilities, $\tilde{p}_{0,0}$ and $\tilde{p}_{\eta,\eta}$, are:
\begin{equation}
    \tilde{p}_{0,0} = \frac{p_{0,0}}{p_{0,0} + (1-p_{0,0})e^{-\lambda_{\chi}}}
    \label{eqn:p00-tilde}
\end{equation}
\begin{equation}
    \tilde{p}_{\eta,\eta} = \frac{p_{\eta,\eta} e^{-\lambda_{\chi}}}{p_{\eta,\eta} e^{-\lambda_{\chi}} + (1-p_{\eta,\eta})}
    \label{eqn:pee-tilde}
\end{equation}
\end{proposition}

The term $\tilde{\eta}_t$ compensates investors for peg risk. When the peg holds ($\chi_t = 0$), investors face the risk of a future de-peg; when the peg is broken ($\chi_t = -\eta$), they face uncertainty about recovery. The risk-adjusted probabilities $\tilde{p}_{0,0}$ and $\tilde{p}_{\eta,\eta}$ incorporate the market price of peg risk $\lambda_{\chi}$.

\subsection{Implied Equilibrium Relationships}

Combining the restrictions from Propositions \ref{prop:eth-staking}--\ref{prop:steth} yields testable equilibrium relationships.

\begin{proposition}[stETH Lending vs.\ ETH Staking]
\label{prop:steth-staking}
In equilibrium:
\begin{equation}
    \psi_t^{stETH} = \gamma^{ETH}_{t} - \gamma^{stETH}_t + \tilde{\eta}_t - \kappa
    \label{eqn:prop4}
\end{equation}
\end{proposition}

This equation states that the stETH lending yield equals the ETH staking yield, minus the portion passed to stETH holders, plus compensation for peg risk, minus staking costs (which stETH holders avoid).

\begin{proposition}[stETH Lending vs.\ ETH Lending]
\label{prop:steth-lending}
In equilibrium:
\begin{equation}
    \psi_t^{stETH} = \psi^{ETH}_{t} - \gamma^{stETH}_t + \tilde{\eta}_t
    \label{eqn:prop5}
\end{equation}
\end{proposition}

This equation states that the stETH lending yield equals the ETH lending yield, minus the stETH staking yield (which stETH lenders also receive), plus compensation for peg risk. Crucially, staking costs $\kappa$ do not appear because neither ETH lenders nor stETH lenders bear these costs directly.

\section{Empirical Analysis}

We test the equilibrium relationships derived in Propositions \ref{prop:steth-staking} and \ref{prop:steth-lending}. Both propositions imply exact relationships between observable yields. We find that the data strongly reject these relationships, indicating market inefficiency.

\subsection{Data}

We obtain daily data from three sources. Lending yields for ETH and stETH are from Aavescan.\footnote{Aavescan: \url{https://aavescan.com/}. Accessed September 2025.} Staking yields are from Dune Analytics.\footnote{Dune query: \url{https://dune.com/queries/570874/1464690}. Accessed September 2025.} The stETH/ETH exchange rate is computed from daily prices on CoinGecko.\footnote{CoinGecko historical price data: \url{https://www.coingecko.com/en/coins/lido-staked-ether} and \url{https://www.coingecko.com/en/coins/ethereum}. Accessed September 2025.} Our sample spans January 30, 2023 to September 21, 2025, yielding 966 daily observations. Table \ref{tab:summary} presents summary statistics:

\begin{table}[H]
\centering
\caption{Summary Statistics}
\label{tab:summary}
\begin{tabular}{lcccc}
\hline\hline
& Mean & Std.\ Dev. & Min & Max \\
\hline
$\psi^{ETH}$ (ETH lending yield) & 1.92\% & 0.65\% & 0.95\% & 17.80\% \\
$\psi^{stETH}$ (stETH lending yield) & 0.05\% & 0.07\% & 0.00\% & 0.67\% \\
$\gamma^{ETH}$ (ETH staking yield) & 3.29\% & 0.64\% & 2.18\% & 12.27\% \\
$\gamma^{stETH}$ (Lido staking yield) & 2.96\% & 0.58\% & 1.97\% & 11.05\% \\
stETH/ETH ratio & 0.999 & 0.001 & 0.990 & 1.015 \\
\hline
Observations & \multicolumn{4}{c}{966} \\
Sample period & \multicolumn{4}{c}{January 30, 2023 -- September 21, 2025} \\
\hline\hline
\end{tabular}
\vspace{0.3cm}
\begin{minipage}{0.9\textwidth}
\footnotesize
\textit{Notes:} Yields are annualized. $\psi^{ETH}$ and $\psi^{stETH}$ are lending yields from Aave. $\gamma^{ETH}$ is the Ethereum protocol staking yield. $\gamma^{stETH}$ is the staking yield passed to Lido stETH holders (after Lido's 10\% fee). The stETH/ETH ratio is the daily average exchange rate.
\end{minipage}
\end{table}

\subsection{Methodology}

Given measurement error in the data, Proposition \ref{prop:steth-staking} implies:
\begin{equation}
    \psi_t^{stETH} = \alpha + \beta \left( \gamma^{ETH}_{t} - \gamma^{stETH}_t \right) + \epsilon_t
    \label{eqn:reg-prop4}
\end{equation}
where $\beta = 1$. Thus, to test market efficiency, we test:
\begin{equation}
    H_0: \beta = 1 \quad \text{versus} \quad H_A: \beta \neq 1
    \label{eqn:hyp-prop4}
\end{equation}

Note that the regressor $\gamma^{ETH} - \gamma^{stETH}$ is the spread between the protocol staking yield and the yield that Lido passes through to stETH holders. Consider an investor deciding between two strategies: direct staking, which earns the full protocol yield $\gamma^{ETH}$, or liquid staking through Lido, which earns the Lido yield $\gamma^{stETH}$ plus whatever can be earned by lending the stETH at Aave, $\psi^{stETH}$. In equilibrium, these strategies must offer comparable returns. If the spread $\gamma^{ETH} - \gamma^{stETH}$ increases, then direct staking becomes more attractive. For liquid staking to remain competitive, the stETH lending yield $\psi^{stETH}$ must rise by the same amount. If it does not, capital will flow out of liquid staking and into direct staking until yields adjust. More formally, this one-for-one relationship is an implication from Proposition \ref{prop:steth-staking}.

Similarly, Proposition \ref{prop:steth-lending} implies:
\begin{equation}
    \psi_t^{stETH} = \alpha + \beta \left( \psi^{ETH}_{t} - \gamma^{stETH}_t \right) + \epsilon_t
    \label{eqn:reg-prop5}
\end{equation}
where $\beta = 1$. Thus, to test market efficiency, we test:
\begin{equation}
    H_0: \beta = 1 \quad \text{versus} \quad H_A: \beta \neq 1
    \label{eqn:hyp-prop5}
\end{equation}

The regressor $\psi^{ETH} - \gamma^{stETH}$ compares the ETH lending yield to the staking component of the liquid staking strategy. Consider an investor choosing between lending ETH at Aave, which earns $\psi^{ETH}$, or pursuing the liquid staking strategy, which earns the Lido staking yield $\gamma^{stETH}$ plus the stETH lending yield $\psi^{stETH}$. If $\psi^{ETH} - \gamma^{stETH}$ increases, then direct ETH lending becomes more attractive relative to the staking component of liquid staking. For the liquid staking strategy to remain competitive, the stETH lending yield must rise by the same amount. More formally, this one-for-one adjustment is an implication from Proposition \ref{prop:steth-lending}.

Since our model allows for arbitrary risk pricing, a rejection of $\beta = 1$ cannot be attributed to mispriced risk. Rather, it implies a failure of market efficiency. We estimate both regressions with heteroskedasticity and autocorrelation consistent (HAC) standard errors using the Newey-West estimator with 10 lags.

\subsection{Results}

Table \ref{tab:regression} presents our results. Both tests strongly reject market efficiency. We discuss each result in turn.

\begin{table}[htbp]
\centering
\caption{Tests of Market Efficiency}
\label{tab:regression}
\begin{tabular}{lcc}
\hline\hline
& Proposition \ref{prop:steth-staking} & Proposition \ref{prop:steth-lending} \\
& $\gamma^{ETH} - \gamma^{stETH}$ & $\psi^{ETH} - \gamma^{stETH}$ \\
\hline
\\[-0.8em]
$\hat{\beta}$ & $-0.228$ & $0.017$ \\
& $(0.043)$ & $(0.005)$ \\[0.5em]
$\hat{\alpha}$ (annualized) & $0.12\%$ & $0.06\%$ \\
& $(0.02\%)$ & $(0.01\%)$ \\[0.5em]
\hline
\\[-0.8em]
$t$-statistic ($H_0: \beta = 1$) & $-28.33$ & $-179.22$ \\
$p$-value & $<0.001$ & $<0.001$ \\[0.3em]
$R^2$ & $0.143$ & $0.120$ \\
Observations & 966 & 966 \\
\hline\hline
\end{tabular}
\vspace{0.3cm}
\begin{minipage}{0.9\textwidth}
\footnotesize
\textit{Notes:} The dependent variable is the daily stETH lending yield $\psi^{stETH}$. HAC standard errors (Newey-West, 10 lags) in parentheses. Under market efficiency, $\beta = 1$. Both tests reject this hypothesis at the 1\% level.
\end{minipage}
\end{table}

For the test of Proposition \ref{prop:steth-staking}, the estimated coefficient is $\hat{\beta} = -0.228$, which is not only significantly different from unity but is negative. This implies that when the staking yield differential $\gamma^{ETH} - \gamma^{stETH}$ increases, the stETH lending yield $\psi^{stETH}$ actually decreases, contrary to the theoretical prediction.

For the test of Proposition \ref{prop:steth-lending}, the estimated coefficient is $\hat{\beta} = 0.017$. While positive, this is far below the predicted value of unity. A one percentage point increase in $\psi^{ETH} - \gamma^{stETH}$ is associated with only a 0.017 percentage point increase in $\psi^{stETH}$, rather than the one-for-one relationship implied by market efficiency.

\section{Conclusion}

This paper demonstrates market inefficiency in cryptoasset markets while allowing for arbitrary risk pricing. We examine three investments: direct staking, lending via decentralized protocols, and liquid staking. Allowing for arbitrary risk pricing, we derive equilibrium restrictions that must hold regardless of how investors price risk. Our empirical tests strongly reject these equilibrium restrictions, implying market inefficiency. Ultimately, our findings suggest the presence of market frictions that impede capital reallocation across investments. The nature of these frictions remains an open question for future research.

\bibliographystyle{jfe}
\bibliography{references}

\begin{appendices}

\section{Proofs}

\begin{lemma}
\label{lem:eth-pricing}
Let $y_t$ be $\mathcal{F}_t$-measurable. If the return $R_{t,t+1} = \exp\{y_t + r^{ETH}_{t,t+1}\}$ satisfies the Euler equation \eqref{eqn:euler}, then:
\begin{equation}
    \mathbb{E}[r^{ETH}_{t,t+1} | \mathcal{F}_t] = r_{t,t+1} + \lambda_{ETH}\sqrt{v_t^{ETH}} - y_t - \frac{v_t^{ETH}}{2}.
    \label{eqn:lemma-eth}
\end{equation}
\end{lemma}

\begin{proof}
The Euler equation \eqref{eqn:euler} requires:
\begin{equation}
    \mathbb{E}\left[\Lambda_{t,t+1} \exp\{y_t + r^{ETH}_{t,t+1}\} ~\Big|~ \mathcal{F}_t\right] = 1.
    \label{eqn:lem-euler}
\end{equation}
Substituting the pricing kernel \eqref{eqn:kernel} and factoring out $\mathcal{F}_t$-measurable terms:
\begin{equation}
    \exp\left\{-r_{t,t+1} - \tfrac{1}{2}\lambda_{ETH}^2 + y_t\right\} \cdot \mathbb{E}\left[\exp\left\{r^{ETH}_{t,t+1} - \lambda_{ETH} \varepsilon^{ETH}_{t+1} - \lambda_{\chi} \omega_{t+1} + \xi_t(\lambda_{\chi})\right\} ~\Big|~ \mathcal{F}_t\right] = 1.
    \label{eqn:lem-factor}
\end{equation}
By the definition of $\varepsilon^{ETH}_{t+1}$, we have $r^{ETH}_{t,t+1} = \mu_t^{ETH} + \sqrt{v_t^{ETH}}\, \varepsilon^{ETH}_{t+1}$. Substituting:
\[
r^{ETH}_{t,t+1} - \lambda_{ETH} \varepsilon^{ETH}_{t+1} = \mu_t^{ETH} + (\sqrt{v_t^{ETH}} - \lambda_{ETH})\varepsilon^{ETH}_{t+1}.
\]
By independence of $\varepsilon^{ETH}_{t+1}$ and $\omega_{t+1}$, the conditional expectation in \eqref{eqn:lem-factor} factors:
\[
    \exp\{\mu_t^{ETH}\} \cdot \mathbb{E}\left[\exp\{(\sqrt{v_t^{ETH}} - \lambda_{ETH})\varepsilon^{ETH}_{t+1}\} ~\Big|~ \mathcal{F}_t\right] \cdot \mathbb{E}\left[\exp\{-\lambda_{\chi} \omega_{t+1} + \xi_t(\lambda_{\chi})\} ~\Big|~ \mathcal{F}_t\right].
\]
Since $\varepsilon^{ETH}_{t+1} | \mathcal{F}_t \sim N(0,1)$, the moment generating function gives:
\[
\mathbb{E}\left[\exp\{(\sqrt{v_t^{ETH}} - \lambda_{ETH})\varepsilon^{ETH}_{t+1}\} ~\Big|~ \mathcal{F}_t\right] = \exp\left\{\tfrac{1}{2}(\sqrt{v_t^{ETH}} - \lambda_{ETH})^2\right\}.
\]
The normalizing constant $\xi_t(\lambda_{\chi})$ is defined such that $\mathbb{E}[\exp\{-\lambda_{\chi} \omega_{t+1} + \xi_t(\lambda_{\chi})\} | \mathcal{F}_t] = 1$. Substituting these results into \eqref{eqn:lem-factor} and taking logarithms:
\[
-r_{t,t+1} - \tfrac{1}{2}\lambda_{ETH}^2 + y_t + \mu_t^{ETH} + \tfrac{1}{2}(\sqrt{v_t^{ETH}} - \lambda_{ETH})^2 = 0.
\]
Solving for $\mu_t^{ETH} = \mathbb{E}[r^{ETH}_{t,t+1} | \mathcal{F}_t]$ gives \eqref{eqn:lemma-eth}.
\end{proof}

\begin{proof}[Proof of Proposition \ref{prop:eth-staking}]
The staking return \eqref{eqn:cumethreturn} has the form $R_{t,t+1} = \exp\{y_t + r^{ETH}_{t,t+1}\}$ with $y_t = \gamma^{ETH}_t - \kappa$. Since $\gamma^{ETH}_t$ and $\kappa$ are $\mathcal{F}_t$-measurable, Lemma \ref{lem:eth-pricing} applies. Substituting $y_t = \gamma^{ETH}_t - \kappa$ into \eqref{eqn:lemma-eth} gives \eqref{eqn:eth-staking}.
\end{proof}

\begin{proof}[Proof of Proposition \ref{prop:eth-lending}]
The lending return \eqref{eqn:cumethreturn-yf} has the form $R_{t,t+1} = \exp\{y_t + r^{ETH}_{t,t+1}\}$ with $y_t = \psi_t^{ETH}$. Since $\psi_t^{ETH}$ is $\mathcal{F}_t$-measurable, Lemma \ref{lem:eth-pricing} applies. Substituting $y_t = \psi_t^{ETH}$ into \eqref{eqn:lemma-eth} gives \eqref{eqn:eth-lending}.
\end{proof}

\begin{lemma}
\label{lem:peg-pricing}
Let $\tilde{\eta}_t$ be defined by \eqref{eqn:peg-premium}. Then:
\begin{equation}
    \mathbb{E}\left[\exp\{\chi_{t+1} - \lambda_{\chi} \omega_{t+1} + \xi_t(\lambda_{\chi})\} ~\Big|~ \mathcal{F}_t\right] = \exp\{\chi_t - \tilde{\eta}_t\}.
    \label{eqn:lemma-peg}
\end{equation}
\end{lemma}

\begin{proof}
We verify \eqref{eqn:lemma-peg} for each state $\chi_t \in \{0, -\eta\}$.

\medskip
\noindent\textit{Case $\chi_t = 0$:} The state $\chi_{t+1} \in \{0, -\eta\}$ with physical probabilities $p_{0,0}$ and $1 - p_{0,0}$. By definition, $\omega_{t+1} = 1$ when $\chi_{t+1} = -\eta$. Evaluating the conditional expectation:
\[
\mathbb{E}\left[\exp\{\chi_{t+1} - \lambda_{\chi} \omega_{t+1} + \xi_t(\lambda_{\chi})\} ~\Big|~ \chi_t = 0\right] = e^{\xi_t(\lambda_{\chi})}\left[p_{0,0} \cdot 1 + (1-p_{0,0}) \cdot e^{-\eta - \lambda_{\chi}}\right].
\]
The normalizing constant satisfies $\mathbb{E}[\exp\{-\lambda_{\chi} \omega_{t+1} + \xi_t(\lambda_{\chi})\} | \chi_t = 0] = 1$, which gives $e^{\xi_t(\lambda_{\chi})} = [p_{0,0} + (1-p_{0,0})e^{-\lambda_{\chi}}]^{-1}$. Substituting:
\[
\mathbb{E}\left[\exp\{\chi_{t+1} - \lambda_{\chi} \omega_{t+1} + \xi_t(\lambda_{\chi})\} ~\Big|~ \chi_t = 0\right] = \frac{p_{0,0} + (1-p_{0,0})e^{-\eta - \lambda_{\chi}}}{p_{0,0} + (1-p_{0,0})e^{-\lambda_{\chi}}}.
\]
From \eqref{eqn:p00-tilde}, define $\tilde{p}_{0,0} = p_{0,0}/[p_{0,0} + (1-p_{0,0})e^{-\lambda_{\chi}}]$. The right-hand side simplifies to $\tilde{p}_{0,0} + (1-\tilde{p}_{0,0})e^{-\eta} = e^{-\eta} + \tilde{p}_{0,0}(1 - e^{-\eta})$. From \eqref{eqn:peg-premium} with $\chi_t = 0$, this equals $\exp\{-\tilde{\eta}_t\} = \exp\{\chi_t - \tilde{\eta}_t\}$.

\medskip
\noindent\textit{Case $\chi_t = -\eta$:} The state $\chi_{t+1} \in \{0, -\eta\}$ with physical probabilities $1 - p_{\eta,\eta}$ and $p_{\eta,\eta}$. By definition, $\omega_{t+1} = 1$ when $\chi_{t+1} = -\eta$. Evaluating the conditional expectation:
\[
\mathbb{E}\left[\exp\{\chi_{t+1} - \lambda_{\chi} \omega_{t+1} + \xi_t(\lambda_{\chi})\} ~\Big|~ \chi_t = -\eta\right] = e^{\xi_t(\lambda_{\chi})}\left[(1-p_{\eta,\eta}) \cdot 1 + p_{\eta,\eta} \cdot e^{-\eta - \lambda_{\chi}}\right].
\]
The normalizing condition gives $e^{\xi_t(\lambda_{\chi})} = [(1-p_{\eta,\eta}) + p_{\eta,\eta}e^{-\lambda_{\chi}}]^{-1}$. Substituting:
\[
\mathbb{E}\left[\exp\{\chi_{t+1} - \lambda_{\chi} \omega_{t+1} + \xi_t(\lambda_{\chi})\} ~\Big|~ \chi_t = -\eta\right] = \frac{(1-p_{\eta,\eta}) + p_{\eta,\eta}e^{-\eta - \lambda_{\chi}}}{(1-p_{\eta,\eta}) + p_{\eta,\eta}e^{-\lambda_{\chi}}}.
\]
From \eqref{eqn:pee-tilde}, define $\tilde{p}_{\eta,\eta} = p_{\eta,\eta}e^{-\lambda_{\chi}}/[(1-p_{\eta,\eta}) + p_{\eta,\eta}e^{-\lambda_{\chi}}]$. The right-hand side simplifies to $(1-\tilde{p}_{\eta,\eta}) + \tilde{p}_{\eta,\eta}e^{-\eta} = 1 - \tilde{p}_{\eta,\eta}(1 - e^{-\eta})$. From \eqref{eqn:peg-premium} with $\chi_t = -\eta$:
\[
\exp\{\chi_t - \tilde{\eta}_t\} = e^{-\eta} \cdot [e^{\eta} - \tilde{p}_{\eta,\eta}(e^{\eta} - 1)] = 1 - \tilde{p}_{\eta,\eta}(1 - e^{-\eta}).
\]
This confirms \eqref{eqn:lemma-peg}.
\end{proof}

\begin{proof}[Proof of Proposition \ref{prop:steth}]
Apply the Euler equation \eqref{eqn:euler} to the stETH return \eqref{eqn:cumstethreturn}:
\[
\mathbb{E}\left[\Lambda_{t,t+1} \exp\{\gamma^{stETH}_t + \psi^{stETH}_t + r^{stETH}_{t,t+1}\} ~\Big|~ \mathcal{F}_t\right] = 1.
\]
Substituting \eqref{eqn:ethsteth} for $r^{stETH}_{t,t+1}$:
\[
\mathbb{E}\left[\Lambda_{t,t+1} \exp\{\gamma^{stETH}_t + \psi^{stETH}_t + r^{ETH}_{t,t+1} + \chi_{t+1} - \chi_t\} ~\Big|~ \mathcal{F}_t\right] = 1.
\]
The quantities $\gamma^{stETH}_t$, $\psi^{stETH}_t$, and $\chi_t$ are $\mathcal{F}_t$-measurable. Substituting the pricing kernel \eqref{eqn:kernel} and factoring:
\begin{equation}
\exp\left\{-r_{t,t+1} - \tfrac{1}{2}\lambda_{ETH}^2 + \gamma^{stETH}_t + \psi^{stETH}_t - \chi_t\right\} \cdot \mathbb{E}\left[\exp\{A_{t+1} + B_{t+1}\} ~\Big|~ \mathcal{F}_t\right] = 1
\label{eqn:pf3-factor}
\end{equation}
where $A_{t+1} := r^{ETH}_{t,t+1} - \lambda_{ETH}\varepsilon^{ETH}_{t+1}$ and $B_{t+1} := \chi_{t+1} - \lambda_{\chi}\omega_{t+1} + \xi_t(\lambda_{\chi})$.

By independence of $\varepsilon^{ETH}_{t+1}$ and $\omega_{t+1}$, and since $\chi_{t+1}$ is determined by $(\chi_t, \omega_{t+1})$:
\[
\mathbb{E}\left[\exp\{A_{t+1} + B_{t+1}\} ~\Big|~ \mathcal{F}_t\right] = \mathbb{E}\left[\exp\{A_{t+1}\} ~\Big|~ \mathcal{F}_t\right] \cdot \mathbb{E}\left[\exp\{B_{t+1}\} ~\Big|~ \mathcal{F}_t\right].
\]
For the first factor, substituting $r^{ETH}_{t,t+1} = \mu_t^{ETH} + \sqrt{v_t^{ETH}}\,\varepsilon^{ETH}_{t+1}$ and applying the moment generating function of $\varepsilon^{ETH}_{t+1} | \mathcal{F}_t \sim N(0,1)$:
\[
\mathbb{E}\left[\exp\{A_{t+1}\} ~\Big|~ \mathcal{F}_t\right] = \exp\left\{\mu_t^{ETH} + \tfrac{1}{2}(\sqrt{v_t^{ETH}} - \lambda_{ETH})^2\right\}.
\]
For the second factor, Lemma \ref{lem:peg-pricing} gives:
\[
\mathbb{E}\left[\exp\{B_{t+1}\} ~\Big|~ \mathcal{F}_t\right] = \exp\{\chi_t - \tilde{\eta}_t\}.
\]
Substituting into \eqref{eqn:pf3-factor}, the terms $-\chi_t$ and $+\chi_t$ cancel. Taking logarithms:
\[
-r_{t,t+1} - \tfrac{1}{2}\lambda_{ETH}^2 + \gamma^{stETH}_t + \psi^{stETH}_t + \mu_t^{ETH} + \tfrac{1}{2}(\sqrt{v_t^{ETH}} - \lambda_{ETH})^2 - \tilde{\eta}_t = 0.
\]
Solving for $\mu_t^{ETH} = \mathbb{E}[r^{ETH}_{t,t+1} | \mathcal{F}_t]$ gives \eqref{eqn:steth}.
\end{proof}

\begin{proof}[Proof of Proposition \ref{prop:steth-staking}]
Equations \eqref{eqn:eth-staking} and \eqref{eqn:steth} imply:
\[
r_{t,t+1} + \lambda_{ETH}\sqrt{v_t^{ETH}} - \gamma^{ETH}_t + \kappa - \frac{v_t^{ETH}}{2} = r_{t,t+1} + \lambda_{ETH}\sqrt{v_t^{ETH}} + \tilde{\eta}_t - \gamma^{stETH}_t - \psi_t^{stETH} - \frac{v_t^{ETH}}{2}.
\]
which is equivalent to \eqref{eqn:prop4}.
\end{proof}

\begin{proof}[Proof of Proposition \ref{prop:steth-lending}]
Equations \eqref{eqn:eth-lending} and \eqref{eqn:steth} imply:
\[
r_{t,t+1} + \lambda_{ETH}\sqrt{v_t^{ETH}} - \psi_t^{ETH} - \frac{v_t^{ETH}}{2} = r_{t,t+1} + \lambda_{ETH}\sqrt{v_t^{ETH}} + \tilde{\eta}_t - \gamma^{stETH}_t - \psi_t^{stETH} - \frac{v_t^{ETH}}{2}.
\]
which is equivalent to \eqref{eqn:prop5}.
\end{proof}

\end{appendices}

\end{document}